\documentclass[conference]{IEEEtran}
\IEEEoverridecommandlockouts

\usepackage{pgf}
\usepackage{url}
\usepackage{algorithm}
\usepackage{algorithmicx}
\usepackage{algpseudocode}
\usepackage{subfigure}
\usepackage{amsmath,amssymb,amsfonts}
\usepackage{xcolor}
\usepackage{color}

\algnewcommand\algorithmicforeach{\textbf{for each}}
\algdef{S}[FOR]{ForEach}[1]{\algorithmicforeach\ #1\ \algorithmicdo}

\algblockdefx{ForAllP}{EndFaP}[1]%
  {\textbf{for all}#1 \textbf{do in parallel}}%
  {\textbf{end for}}
\algtext*{EndFaP}%

\algblockdefx{ForAllNP}{EndFaP}[1]%
  {\textbf{for all}#1 \textbf{do serially}}%
  {\textbf{end for}}
\algtext*{EndFaP}%

\algnewcommand\algorithmicswitch{\textbf{switch}}
\algnewcommand\algorithmiccase{\textbf{case}}
\algnewcommand\algorithmicassert{\texttt{assert}}
\algnewcommand\Assert[1]{\State \algorithmicassert(#1)}%
\algdef{SE}[SWITCH]{Switch}{EndSwitch}[1]{\algorithmicswitch\ #1\ \algorithmicdo}{\algorithmicend\ \algorithmicswitch}%
\algdef{SE}[CASE]{Case}{EndCase}[1]{\algorithmiccase\ #1}{\algorithmicend\ \algorithmiccase}%
\algtext*{EndSwitch}%
\algtext*{EndCase}%

\newtheorem{proposition}{\textbf{Proposition}}
\newtheorem{proof}{\textbf{Proof}}

\newcommand{\rs}{RadixSpline}
\newcommand{\rmi}{RMI}
\newcommand{\grmi}{GRMI}
\newcommand{\bgrmi}{BGRMI}

\newcommand{\alex}{ALEX}

\newcommand{\amac}{AMAC}

\newcommand{\swwc}{SWWC}
\newcommand{\cdf}{CDF}
\newcommand{\numa}{NUMA}
\newcommand{\simd}{SIMD}

\newcommand{\sosd}{SOSD}
\newcommand{\fb}{\textit{face}}
\newcommand{\osm}{\textit{osm}}
\newcommand{\wiki}{\textit{wiki}}
\newcommand{\seqh}{\textit{seq\_h}}
\newcommand{\unif}{\textit{unif}}
\newcommand{\lognorm}{\textit{lognorm}}

\newcommand{\inlj}{INLJ} 
\newcommand{\linlj}{G{\rmi}-INLJ} 
\newcommand{\blinlj}{BG{\rmi}-INLJ} 
\newcommand{\rmiinlj}{{\rmi}-INLJ} 
\newcommand{\csstreeinlj}{CSS-INLJ} 
\newcommand{\bcsstreeinlj}{BCSS-INLJ} 
\newcommand{\arttreeinlj}{ART-INLJ} 
\newcommand{\hashinlj}{HASH-INLJ} 
\newcommand{\cuckooinlj}{CUCKOO-INLJ} 
\newcommand{\rshashinlj}{RSHASH-INLJ} 
\newcommand{\alexinlj}{ALEX-INLJ} 

\newcommand{\sj}{SJ} 
\newcommand{\lsj}{L-SJ} 
\newcommand{\blsj}{BL-SJ} 
\newcommand{\mpsm}{MPSM} 
\newcommand{\mway}{MWAY} 
\newcommand{\mpsmsj}{MPSM-SJ} 
\newcommand{\mpsmlssj}{MPSM-LS-SJ} 
\newcommand{\mwaysj}{MWAY-SJ} 
\newcommand{\mwaylssj}{MWAY-LS-SJ} 

\newcommand{\hj}{HJ} 
\newcommand{\lpj}{LP-HJ} 
\newcommand{\lnpj}{LNP-HJ} 
\newcommand{\pj}{PJ} 
\newcommand{\npj}{NPJ} 
\newcommand{\pjhj}{P-HJ} 
\newcommand{\npjhj}{NP-HJ} 

\def\BibTeX{{\rm B\kern-.05em{\sc i\kern-.025em b}\kern-.08em
    T\kern-.1667em\lower.7ex\hbox{E}\kern-.125emX}}

\begin{document}

\title{The Case for Learned In-Memory Joins}

\author{\IEEEauthorblockN{Ibrahim Sabek}
\IEEEauthorblockA{\textit{MIT CSAIL} \\
sabek@mit.edu}
\and
\IEEEauthorblockN{Tim Kraska}
\IEEEauthorblockA{\textit{MIT CSAIL} \\
kraska@mit.edu}
}

\maketitle

\begin{abstract}
\label{sec:abstract}
In-memory join is an essential operator in any database engine. It has been extensively investigated in the database literature. In this paper, we study whether exploiting the {\cdf}-based learned models to boost the join performance is practical or not. To the best of our knowledge, we are the first to fill this gap. We investigate the usage of {\cdf}-based partitioning and learned indexes (e.g., Recursive Model Index ({\rmi}) and {\rs}) in the three join categories; indexed nested loop join ({\inlj}), sort-based joins ({\sj}) and hash-based joins ({\hj}). Our study shows that there is a room to improve the performance of {\inlj} and {\sj} categories through our proposed optimized learned variants. Our experimental analysis showed that these proposed learned variants of {\inlj} and {\sj} consistently outperform the state-of-the-art techniques. 

\end{abstract}


\section{Introduction}
\label{sec:introduction}

In-memory join is an essential operator in any database engine. It has been extensively investigated in the database literature (e.g.,~\cite{HYF+08, HL06, AKN12, BAT+13, TAB+13, BLP11, LLA+15, BGN21}). Basically, in-memory join algorithms can be categorized into three main categories:~nested loop joins (e.g.,~\cite{HYF+08, HL06}), whether non-indexed or indexed, sort-based joins (e.g.,~\cite{AKN12, BAT+13}), and hash-based joins (e.g.,~\cite{TAB+13, BLP11, LLA+15, BGN21}). Many recent studies (e.g.,~\cite{BAT+13, SCD16}) have shown that the design, as well as the implementation details, have a substantial impact on the join performance, specially on the modern hardware.

Meanwhile, there has been growing interest in exploiting learned models, such as {\cdf}-based partitioning~\cite{KBC+18, KVC+20} and Recursive Model Indexes ({\rmi})~\cite{KBC+18}, to enhance or replace traditional data structures and algorithms, such as indexing~\cite{KBC+18, KMR+20, FV20, DMY+20, PRK+20, NDA+20}, sorting~\cite{KVC+20, KVK21}, and hashing~\cite{SVH+21, Kra21}. These learned data structures and algorithms can outperform their traditional counterparts on practical workloads as they are better in exploiting trends in the input data and instance-optimizing the performance. Along this line of work, the idea of using {\cdf}-based learned models to improve the performance of in-memory join algorithms was suggested in~\cite{KBC+18}. However, this idea was neither fully investigated nor supported with an experimental evidence.

In this paper, we aim to study, in details, whether exploiting the {\cdf}-based learned models to boost the performance of in-memory joins is a beneficial idea or not. In particular, we investigate the performance of the three join categories; indexed nested loop join ({\inlj}), sort-based joins ({\sj}) and hash-based joins ({\hj}), while modifying or replacing their different phases (e.g., indexing, sorting, joining) with {\cdf}-based and {\rmi}-based variants. Our paper has the following three main contributions:

\noindent\textbf{Investigating Alternatives of Using Learned Models for Joins.} Although the learned model ideas have been introduced in previous works, none of them has been studied in the context of join processing. Therefore, for each join category, we first discuss the straightforward alternatives of directly integrating learned models with its phases. For example, in {\inlj}, {\rmi} can be used to replace the built index on the indexed relation. In {\sj}, the LearnedSort algorithm~\cite{KVC+20} can be used to replace the sorting phase. In {\hj}, {\rmi} or {\cdf}-based partitioning function can be used to replace the hash function that is used to build the hash table. While investigating the different alternatives, we come up with two interesting observations. First, it is hard to use learned models to improve over classical hash join methods~\cite{TAB+13, LLA+15, BGN21}. Second, using learned models "as-is" in replacing {\inlj} and {\sj} phases is sub-optimal. For example, using the LearnedSort algorithm~\cite{KVC+20} as a black-box replacement for the sorting algorithms used in the state-of-the-art {\sj} techniques, such as {\mpsm}~\cite{AKN12} and {\mway}~\cite{TAB+13}, leads to repeating unnecessary work, and hence increases the overall execution time. Another example is replacing the hash functions used in hash-based {\inlj} with {\rmi} models. Unfortunately, this also leads to a poor performance (i.e., high overall latency) as {\rmi} requires a significant overhead to correct its mispredictions.

\noindent\textbf{Optimized Learned Join Variants.} To overcome the performance limitations of using learned models as black-boxes, we introduce optimized variants of the learned join algorithms in the {\inlj} and {\sj} join categories, which lead to several performance improvements. In particular, in {\inlj} category, we introduced two efficient variants of learned {\inlj}. The first variant, namely Gapped {\rmi} {\inlj}, employs an efficient read-only {\rmi} that is built using the concept of \textit{model-based insertion} introduced in~\cite{DMY+20}. To further improve the performance of this variant, we introduced a \textit{cache-optimized hierarchical buffering} technique, namely Requests Buffer, to temporarily store the probe (i.e., lookup) operations over the {\rmi}, which reduces the cache misses, and hence improves the join throughput. We also provided a cache complexity analysis for this buffering technique, and proved its matching to the caching analytical bounds in~\cite{HL06}. The second variant, namely learned-based hash {\inlj}, employs {\rs} to build a hash index that is useful in some scenarios, like joining sequential datasets. In {\sj} category, we proposed a new sort-based join, namely {\lsj}, that exploits a {\cdf}-based model to improve the different {\sj} phases (e.g., sort and join) by just scaling up its prediction degree. {\lsj} has the following two merits. First, it combines {\cdf}-based partitioning with AVX sort~\cite{BAT+13} to provide a new efficient sorting phase. We also provided a complexity analysis for this new sorting phase. Second, it employs the {\cdf}-based model to reduce the amount of redundant join checks in the merge-join phase.

\noindent\textbf{Extensive Experimental Evaluation.} To achieve a deeper understanding of the practicality of using learned models with joins, we conducted a detailed evaluation study for both learned and non-learned variants in the three join categories, using different real ({\sosd}~\cite{MKR+20}) and synthetic datasets. In particular, we compare the optimized learned join variants against, seven {\inlj} baselines (three off-shelf learned ({\rmi}~\cite{KBC+18}, {\rs}~\cite{KMR+20}, {\alex}~\cite{DMY+20}), two tree-based (Cache-Sensitive Search Trees~\cite{HYF+08}, Adaptive Radix Trie~\cite{LKN13}), and two hash-based (bucket chaining and Cuckoo hash) indexes), two {\sj} baselines ({\mpsm}~\cite{AKN12}, {\mway}~\cite{BAT+13}), and two {\hj} baselines (non-partitioned~\cite{TAB+13}, partitioned~\cite{SCD16}). We show that, the two optimized variants of learned {\inlj}, using Gapped {\rmi}, can provide better performance than almost all traditional baselines (e.g., 2.7X faster than hash-based {\inlj}) in many scenarios and with different datasets. In addition, the optimized learned {\sj} variant can be 2X faster than {\mway}~\cite{BAT+13}, which is the best sort-based join baseline in the literature. We also extend this analysis to duplicate datasets, multiple dataset sizes and ratios, skewness and parallelism degrees, and other performance counters (e.g., cache misses, branch misses and instructions) as well as hardware aspects (e.g., NUMA awareness and page sizes).

We believe our study is beneficial for researchers and practitioners, who are interested in building instance-optimized DBMS, like SageDB~\cite{KAB+19}, as it provides good practices for implementing efficient learned join algorithms. 



\section{Background}
\label{sec:background}


\begin{figure*}
    \centering
    \begin{minipage}{0.8\textwidth}
        \centering
        \includegraphics[width=0.9\textwidth, height=0.2\textheight]{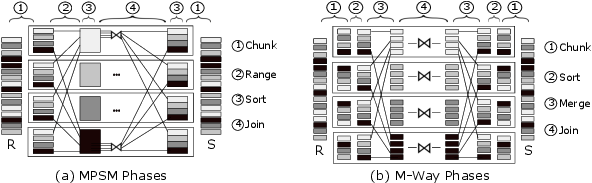}
        \caption{Sort-based Joins: {\mpsm} and {\mway}}
        \label{fig:background_mpsm_mway}
    \end{minipage}%
    \begin{minipage}{0.2\textwidth}
        \centering
        \includegraphics[height=0.18\textheight]{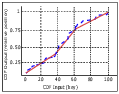}
        \caption{Linear Models Approximation of {\cdf}}
        \label{fig:background_cdf}
    \end{minipage}
\end{figure*}



In this section, we provide a brief background on the different join categories that we study (section~\ref{sec:background_inlj} to~\ref{sec:background_hj}), and the learned concepts and indexes we exploit (Section~\ref{sec:background_rmi_cdf_partitioning}).


\subsection{Indexed Nested Loop Joins ({\inlj})}
\label{sec:background_inlj}
This is the most basic type of join. Here, we assume that one input relation, say $R$, has an index on the join key. Then, the algorithm simply scans each tuple in the second relation $S$, uses the index to fetch the tuples from $R$ with the same join key, and finally performs the join checks. {\inlj} achieves a good performance when using a clustered index as the tuples with the same join key will have a spatial locality, and hence improves the caching behavior.

\subsection{Sort-based Joins ({\sj})}
\label{sec:background_sj}

The main idea of sort-based join is to first sort both input relations on the same join key. Then, the sorted relations are joined using an efficient merge-join algorithm to find all matching tuples. Here, we describe the details of two efficient multi-core variants of sort-based joins, {\mpsm}~\cite{AKN12} and {\mway}~\cite{BAT+13}, as follows: 

\noindent\textbf{Massively Parallel Sort-Merge ({\mpsm}).} The idea of {\mpsm} is simple. It generates sorted runs for the different partitions of the input relations in parallel. Then, these sorted runs are joined directly, without the need to merge each relation entirely first. The steps of a range-partitioned variant of {\mpsm} from~\cite{AKN12}, which is very efficient for {\numa}-aware systems, are as follows: (1)~chunk each input relation into equi-sized chunks among the workers (\textit{Chunk} step), (2)~globally range-partition the smaller relation, say $R$, such that different ranges of $R$ are assigned to different workers (\textit{Range} step), (3)~locally sort each partition from $R$ and $S$ on its own worker (\textit{Sort} step) (note that after this step $R$ is globally sorted, but $S$ is partially sorted), and (4)~merge-join each sorted run from $R$ with all the sorted runs from $S$ (\textit{Join} step).


\noindent\textbf{Multi-Way Sort-Merge ({\mway}).} In {\mway}, the algorithm sorts both relations, merges the sorted runs from each relation, and finally applies a one-to-one merge-join step. The steps of the {\mway} algorithm from~\cite{BAT+13} are as follows: (1)~chunk each input relation into equally-sized chunks among the workers, where each chunk is further partitioned into few segments using efficient radix partitioning (\textit{Chunk} step), (2)~locally sort all segments of each partition from $R$ and $S$ on its own worker (\textit{Sort} step), (3)~multi-way merge all segments of each relation such that different ranges of the merged relation are assigned to different workers (\textit{Merge} step), and (4)~locally merge-join each partition from merged $R$ with exactly one corresponding partition from merged $S$ that is located on the same worker (\textit{Join} step). Both sorting and merging steps are implemented with {\simd} bitonic sorting and merge networks, respectively.


\subsection{Hash-based Joins ({\hj})}
\label{sec:background_hj}

\noindent\textbf{Non-Partitioned Hash Join ({\npj}).} It is a direct parallel variation of the canonical hash join~\cite{LLA+15, TAB+13}. Basically, the algorithm chunks each input relation into equi-sized chunks and assigns them to a number of worker threads. Then, it runs in two phases: \textit{build} and \textit{probe}. In the \textit{build} phase, all workers use the chunks of one input relation, say $R$, to concurrently build a single global hash table. In the \textit{probe} phase, after all workers are done building the hash table, each worker starts probing its chunk from the second relation $S$ against the build hash table. Typically, the build and probe phases benefit from simultaneous multi-threading and out-of-order execution to hide cache miss penalties automatically. 

\noindent\textbf{Partitioned Hash Join ({\pj}).} The main idea is to improve over the non-partitioned hash join by chunking the input relations into small pairs of chunks such that each chunk of one input relation, say $R$, fits into the cache. This will significantly reduce the number of cache misses when building and probing hash tables. The state-of-the-art algorithm proposed in~\cite{TAB+13} has two main phases: \textit{partition} and \textit{join}. In the \textit{partition} phase, a multi-pass radix partitioning (usually 2 or 3 passes are enough) is applied on the two input relations. The key idea is to keep the number of partitions in each pass smaller than the number of TLB entries to avoid any TLB misses. The algorithm builds a histogram to determine the output memory ranges of each target partition. Typically, one global histogram is built. However, we adapt an efficient variation of the partitioning phase from ~\cite{SCD16} that builds multiple local histograms and uses them to partition the input chunks locally to avoid the excessive amount of non-local writes. Other optimization tricks, such as software write-combine buffers ({\swwc}) and streaming~\cite{TAB+13}, are applied as well. In the \textit{join} phase, each worker runs a cache-local hash join on each pair of final partitions from $R$, and $S$.

\subsection{{\cdf}-based Partitioning and Learned Indexing Structures}
\label{sec:background_rmi_cdf_partitioning}

According to statistics, the \textit{Cumulative Distribution Function ({\cdf})} of an input key $x$ is the proportion of keys less than $x$ in a sorted array $A$. Recently,~\cite{KBC+18} has proposed the idea of using {\cdf} as a function to map keys to their relative positions in a sorted index. In particular, the position of each key $x$ can be simply predicted as $CDF(x)*|A|$ where $|A|$ is the index length. Note that, in case of having an exact empirical {\cdf}, this operation will result into placing each key in its exact position in $A$ without error (i.e., sorting $A$). However, obtaining this exact {\cdf} is impractical. Typically, the {\cdf} of input data is estimated in three steps: (1)~sample the input data, (2)~sort the sample, (3)~approximate the {\cdf} of the sorted sample using learned models (e.g., linear regression). Figure~\ref{fig:background_cdf} shows an example on {\cdf} approximation using linear models.

\noindent\textbf{{\cdf}-based Partitioning.} {\cdf} can be used to perform logical partitioning (i.e., bucketization) of the input keys as follows: given an input key $x$, the partition index of each key $x$ is determined as $CDF(x)*|P|$, where $|P|$ is the number of partitions, instead of the actual array length $|A|$ (typically, $|P|$ is much less than $|A|$). Such {\cdf}-based partitioning can be done in a single pass over the data since keys are not required to be exactly sorted within each partition. However, the output partitions (i.e., buckets) are still relatively sorted, and the same {\cdf} can be used to recursively partition the keys within each bucket, if needed, by just scaling up the number of partitions $|P|$. In this paper, we exploit the {\cdf}-based partitioning to improve the performance of sort-based joins as described later in Section~\ref{sec:lsj}

\noindent\textbf{Recursive Model Index ({\rmi}).} The \textit{learned index} idea proposed in~\cite{KBC+18} is based on an insight of rethinking the traditional B+Tree indexes as learned models (e.g., {\cdf}-based linear regression). So, instead of traversing the tree based on comparisons with keys in different nodes and branching via stored pointers, \textit{a hierarchy of models}, namely Recursive Model Index ({\rmi}), can rapidly predict the location of a key using fast arithmetic calculations such as additions and multiplications. Typically, {\rmi} has a static depth of two or three levels, where the root model gives an initial prediction of the {\cdf} for a specific key. Then, this prediction is recursively refined by more accurate models in the subsequent levels, with the leaf-level model making the final prediction for the position of the key in the data structure. As previously mentioned, the models are built using approximated {\cdf}, and hence the final prediction at the leaf-level could be inaccurate, yet error-bounded. Therefore, {\rmi} keeps min and max error bounds for each leaf-level model to perform a local search within these bounds (e.g., binary search), and obtain the exact position of the key. Technically, each model in {\rmi} could be of a different type (e.g., linear regression, neural network). However, linear models usually provide a good balance between accuracy and size. A linear model needs to store two parameters (intercept and slope) only, which have low storage overhead. 

\noindent\textbf{RadixSpline (RS).} {\rs}~\cite{KMR+20} is a 2-levels learned index that followed {\rmi}. It is built in a bottom-up fashion with a single pass over the input keys, where it consists of a set of spline points that are used to approximate the {\cdf} of input keys, and an on-top radix table that indexes $r$-bit prefixes of these spline points. 

A {\rs} lookup operation can be broken down into four steps: (1)~extract the $r$ most significant bits $b$ of the input key $x$, (2)~retrieve the spline points at offsets $b$ and $b+1$ from the radix table, (3)~apply a binary search among the retrieved spline points to locate the two spline points that approximate the {\cdf} of the input key $x$, and (4)~perform a linear interpolation between the retrieved spline points to predict the location of the key $x$ in the underlying data structure. Similar to {\rmi}, {\rs} is based on an approximated {\cdf}, and hence might lead to a wrong prediction. To find the exact location of key $x$, a typical binary search is performed within a small error bounds around the final {\rs} prediction. These error bounds are user-defined and can be tuned beforehand.

According to a recent benchmark study~\cite{MKR+20}, both {\rmi} and {\rs} are considered the most efficient read-only learned indexing structures. Therefore, in this paper, we explore their usage in improving the performance of indexed nested loop joins as described later in Section~\ref{sec:linlj}.

\section{Learned Indexed Nested Loop Join}
\label{sec:linlj}

In this section, we explore using the learned index idea in indexed nested loop joins ({\inlj}). One straightforward approach is to index one input relation, say $R$, with an {\rmi} and then use the other relation, say $S$, to probe the built {\rmi} and check for the join matches. This is an appealing idea because {\rmi}~\cite{KBC+18} typically encodes the built index as few linear model parameters (i.e., small index representation), and in turn improves the {\inlj} cache behavior. However, our experiments showed that using typical {\rmi}s only is not the best option in all cases. This is mainly due to the significant overhead of the "last-mile" search (i.e., correcting the {\rmi} mispredictions using binary search) when having challenging datasets. Therefore, we first propose a new variation of {\rmi} (Section~\ref{sec:linlj_gapped_rmi}) that substantially improves the "last-mile" search performance for general lookup queries (e.g., filtering), yet, with higher space requirements (i.e., less cache-friendly). Then, we provide a customized variation of this index (Section~\ref{sec:linlj_optimized_buffer}) for {\inlj} that exploits a cache-optimized buffering technique to increase the caching benefits while keeping the improved search performance during the {\inlj} processing. \textcolor{black}{Finally, we discuss the idea of using the {\cdf}-based models as a hash function to build and probe a typical hash table in the hash-based {\inlj} (Section~\ref{sec:linlj_rmi_hash})}.

\subsection{Gapped {\rmi} ({\grmi})}
\label{sec:linlj_gapped_rmi}

\begin{figure}
    \begin{center}
    \includegraphics[width=3.2in]{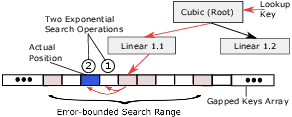}
    \end{center}
    \caption{Example of 2-levels Gapped {\rmi} ({\grmi})}
    \label{fig:linlj_grmi}
\end{figure}

The performance of any {\inlj} algorithm mainly depends on the used index. However, when using {\rmi} in {\inlj} over challenging inputs, the join throughput deteriorates as the input data is difficult to learn and the accuracy of built {\rmi} becomes poor. In this case, the {\rmi} predicts a location which is far away from the true one, and hence the overhead of searching between the error bounds becomes significant. This is because when an {\rmi} is built for an input relation, it never changes the positions of the tuples in this relation. 
A key insight from a recent updatable learned index~\cite{DMY+20}, namely {\alex}, shows that when tuples are inserted using a learned model, the search over these tuples later, using the same model, will be improved significantly. 
Basically, \textit{model-based} insertion places a key at the position where the model predicts that the key should be. As a result, \textit{model-based} insertion reduces the model misprediction errors. 
This inspires us to propose a new read-only {\rmi}-based index, namely Gapped {\rmi} ({\grmi}), that substantially improves the lookup performance over typical {\rmi}, specially for datasets with challenging distributions.   
 
\noindent\textbf{{\grmi} Structure.} The proposed {\grmi} consists of a \textit{typical {\rmi}} and \textit{two gapped arrays}, one for keys and one for payloads. We assume both keys and payloads have fixed sizes. Any payload could be an actual record or a pointer to a variable-sized record, stored somewhere else. The gapped arrays are filled by model-based inserts of the input tuples (described later). The extra space used for the gaps is distributed among the elements of the array, to ensure that tuples are located closely to the predicted position when possible. The number of allocated gaps per input key is a user-defined parameter that can be tuned. For efficient lookup, the gapped arrays are associated with a bitmap to indicate whether each location in the array is filled or is still a gap. Figure~\ref{fig:linlj_grmi} depicts an example of a 2-levels {\grmi}, where we show only the gapped array for keys (white cells are gaps (i.e., empty)).

{\grmi} has two main differences from {\alex}. First, {\grmi} is a \textit{read-only} data structure which is built once in the beginning and then used for lookup operations only. In contrast, {\alex} is more dynamic and can support updates, insertions and deletes. Second, {\alex} is a B-tree-like data structure which has internal nodes containing pointers and storing intermediate gapped arrays. Unlike {\alex}, {\grmi} just has only two \textit{continuous} gapped arrays, for keys and payloads, to perform the "last-mile" search operations in an efficient manner. Note that {\grmi} employs a typical {\rmi} for predicting the key's location, without any traversal for internal nodes as in {\alex}. Therefore, the {\grmi} lookup is extremely efficient compared to {\alex} as shown in our experiments (Section~\ref{sec:evaluation}).  

\noindent\textbf{{\grmi} Building.} Building a {\grmi} has two straightforward steps. First, a typical {\rmi} is constructed for the input keys. Then, the built {\rmi} is used to model-based insert all tuples from scratch in the gapped arrays. \textit{Note that, in {\inlj}, the index is assumed to be existing beforehand}. This is a typical case in any DBMS where indexes (e.g., hash index, B-tree) are general-purpose and built once to serve many queries (e.g., select, join).  \textit{Therefore, the {\grmi} building overhead has no effect on the {\inlj} performance}. 

\noindent\textbf{{\grmi} Lookup.} Given a key, the {\rmi} is used to predict the location of that key in the gapped array of keys, then, if needed, an \textit{exponential search} is applied till the actual location is found. If a key is found, we return the corresponding payload at the same location in the payloads array, otherwise return null. We use the exponential search to find the key in the neighborhood because, in {\grmi}, the key is likely to be very close to the predicted {\rmi} location, and in this case, the exponential search can find this key with less number of comparison checks compared to both sequential and binary searches. Figure~\ref{fig:linlj_grmi} presents a lookup example, where red arrows show the lookup flow. In this example, the {\rmi} prediction is corrected by two exponential search steps only.

\subsection{Buffered {\grmi} {\inlj}}
\label{sec:linlj_optimized_buffer}

\begin{figure}
    \begin{center}
    \includegraphics[width=3.2in]{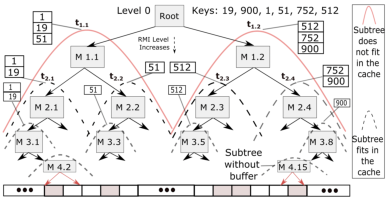}
    \end{center}
    \caption{\textcolor{black}{Example on hierarchical Request Buffers at some {\rmi} models}.}
    \label{fig:linlj_buffer}
\end{figure}

Although leveraging the model-based insertion and exponential search ideas in {\grmi} reduces the number of comparisons and their corresponding cache misses in the "last-mile" search, directly using {\grmi} in {\inlj} is sub-optimal as it introduces another kind of cache misses when iteratively looking up for different keys. Since  keys are stored in quite large gapped arrays in {\grmi}, the probability of finding different keys within the same cacheline is low. To address this issue, we propose to buffer (i.e., temporarily store) the requests for query keys that will be probably accessed from the same part in the \textit{gapped array} (i.e., keys that predict the same models across the {\rmi} levels), and then answer them as a batch. This significantly improves the temporal locality of the index, reduces cache misses, and in turn increases the join throughput.
In this section, we show how we adapt an efficient variation of \textit{hierarchical} buffering techniques~\cite{ZR03, HL06} to improve the throughput of {\inlj} using {\grmi}. We refer to this variation as \textit{Buffered {\grmi}} ({\bgrmi}).

\noindent\textbf{Key Idea.} We view the {\rmi} as a tree with multiple sub-trees rooted at the child models, and the root of each sub-tree $t$ is associated with a fixed-size buffer. We refer to this buffer as \textit{request buffer}, where it temporarily stores all keys that will be predicted using the corresponding child model. \textcolor{black}{Any sub-tree~$t$ can be recursively decomposed into smaller sub-trees at its own child models. Figure~\ref{fig:linlj_buffer} shows an example of sub-trees at different levels of an {\rmi},} where the root of each sub-tree has a request buffer. \textcolor{black}{In this example, the "Root" model has two sub-trees $t_{1.1}$ and $t_{1.2}$ rooted at its child models $M 1.1$ and $M 1.2$, respectively, where for example, the sub-tree $t_{1.1}$ itself has two child sub-trees $t_{2.1}$ and $t_{2.2}$ (with their buffers) rooted at the $M 2.1$ and $M 2.2$ models, respectively, and so on.}

The request buffers for all sub-trees of the indexed relation, say~$R$, are created before the {\inlj} begins. During the {\inlj} operation, the query keys from the non-indexed relation, say $S$, are distributed to the buffers based on the {\rmi} predictions, until a buffer becomes full. When a buffer is full, we flush it by distributing the buffered query keys to its child buffers recursively \textcolor{black}{(e.g., in Figure~\ref{fig:linlj_buffer}, when the buffer of sub-tree $t_{1.1}$ is filled with the query keys $1$, $19$, and $51$, the model $M 1.1$ is then used to distribute these keys to the buffers of its child sub-trees $t_{2.1}$ and $t_{2.2}$, and so on)}. When a query key reaches the leaf-level, the final {\rmi} prediction is performed. When no more query keys are issued, the whole request buffers are flushed in the depth-first order similar to~\cite{HL06}. 

During the rest of this section, \textcolor{black}{we discuss how we define the size of the buffer at the root of any sub-tree $t$, and how the total size of all buffers needed for an {\rmi} satisfies the analytical bound of the \textit{cache complexity} in hierarchical indexes (e.g., trees or {\rmi})~\cite{HL06}}. The cache complexity is defined by the asymptotic number of block transfers between the cache and the memory incurred by a query key.


\noindent\textcolor{black}{\underline{Notation.} For an {\rmi} with a root model $r$, let $n$ be the total number of intermediate and leaf models in this {\rmi} (without the root $r$), $m$ be the fan-out (i.e., number of children per any model in this {\rmi}), and $w$ be the fixed-size, in bytes, of a single model in this {\rmi}. For any sub-tree $t$ in this {\rmi}, let $v_t$ be the size of the buffer associated with the "root" model of $t$, and $g_t$ be the number of "non-root" models in $t$. We also define a function $S(g_t)$ that estimates the total size of buffers associated with the $g_t$ non-root models in any sub-tree $t$.}

\noindent\textcolor{black}{\textbf{Estimating the Buffer Size $v_t$ at the Root of Sub-tree $t$.} The buffer associated with the root of any sub-tree $t$ should consider the following two points. \underline{First}, the expected number of query keys, that will be processed through $t$ and reside in the buffer at its root, depends on the number of the non-root models in $t$ (i.e., $v_t \propto g_t$). Note that $g_t$ is determined by the level of $t$'s root in {\rmi}. Therefore, the closer the level of $t$'s root to the {\rmi} root, the larger buffer size $v_t$ this $t$ needs (e.g., in Figure~\ref{fig:linlj_buffer}, the number of buffered keys in $t_{1.1}$ is larger than in either $t_{2.1}$ or $t_{2.2}$ as the level of $t_{1.1}$'s root is closer to the {\rmi} root). \underline{Second}, the buffer size $v_t$ at $t$'s root should be equal to (or larger than) the total size of $t$'s child sub-trees and their associated buffers (i.e., $v_t \geq S(g_t)$). This is intuitive because if the child sub-trees and all their buffers fit into the cache, they can be reused before eviction from the cache. Therefore, based on these two points, we set $v_t$ to be $wg_t + S(g_t)$, where $wg_t$ and $S(g_t)$ are the total sizes of all the $g_t$ non-root models in $t$ and their associated buffers.}

 
\noindent\textbf{Total Buffers' Size Upper Bound.} Here, we present a worst-case bound on the total size of the request buffers needed by an {\rmi}. Note that our goal is to maximize the {\inlj} performance with {\rmi} and its variant {\grmi}, not to minimize the buffer sizes. Therefore, this analysis is useful for gaining intuition about request buffer sizes, but does not reflect worst-case guarantees in practice.

\begin{proposition}
\textit{Given an {\rmi} with number of intermediate and leaf models $n$ (without the root), child fan-out $m$ and model size $w$, the total size, in bytes, of the request buffers at non-root models of this {\rmi} (i.e., $S(n)$) is $O(n \log_m n)$.}
\label{prop:request_buffer_size}
\end{proposition}

\begin{proof}
As previously mentioned, an {\rmi} with fan-out of $m$ can be viewed as a tree that starts from the root model $r$. Such tree consists of $m$ sub-trees; each sub-tree starts from one of the $m$ child models of $r$. Since the whole {\rmi} tree contains $n$ non-root models, then each of these sub-trees roughly contains $\frac{n}{m}$ nodes. From our definitions, the total size of buffers needed for each sub-tree $t$, without the buffer needed for the root of~$t$, is $S({\frac{n}{m}})$. In addition, the buffer size at the root of~$t$ (i.e., $v_t$) is $w\frac{n}{m} + S(\frac{n}{m})$. Hence, we can redefine $S(n)$ in terms of $S(\frac{n}{m})$ to be as follows:

\begin{equation}
S(n) = m (w\frac{n}{m} + 2S(\frac{n}{m})) = wn + 2m S(\frac{n}{m}) 
\label{eq:buffer_size}
\end{equation}

By solving this equation recursively, we end up with $S(n) = wn (1 + 2 + \ldots + \log_m n) + (2m)^{\log_m n}$, where $(1 + 2 + \ldots + \log_m n)$ is a geometric series. After eliminating lower-order terms of this series and assuming $n \gg m$, $S(n)$ becomes $O(n \log_m n)$ ($w$ is omitted as it is constant). \\

\end{proof}

\textcolor{black}{The base case of Equation~\ref{eq:buffer_size} occurs at the leaf-level (i.e., $S(0) = 0$). Note that we can estimate the value of $S(g_t)$ at any sub-tree $t$ by replacing $n$ in Equation~\ref{eq:buffer_size} with $g_t$, and recursively solve the equation starting from the {\rmi} level of the $t$'s root till the leaf}. Also, note that the assumption of $n \gg m$ (used in the proof) is valid as, in practice, the {\rmi} optimizer~\cite{MZK20} typically models a large input relation with complex distribution using {\rmi} of 2 to 3 levels (excluding the root), and a fan-out $m$ of 1000. In the case of 2 levels, for example, the $\frac{n}{m}$ ratio becomes 1000, which is relatively large.

\noindent\textcolor{black}{\textbf{Optimizing the Number of Buffers in {\rmi}.} Assigning a buffer for "every" sub-tree in {\rmi} is unnecessary and will result in cache misses when accessing the buffers themselves, specially in case of having large {\rmi} (i.e., large number of buffers). Therefore, as an optimization, we keep assigning the buffers on the {\rmi} sub-trees, level by level (starting from level 1), until a sub-tree $t$ and its buffers (excluding the buffer of this sub-tree's root) totally fit into the cache and then stop. For example, in Figure~\ref{fig:linlj_buffer}, we start by assigning buffers to sub-trees $t_{1.1}$ and $t_{1.2}$ ({\rmi} level 1). Then, we assign buffers to their child sub-trees $t_{2.1}, t_{2.2}, t_{2.3}$ and $t_{2.4}$ ({\rmi} level 2). Here, we assume that $t_{2.1}, t_{2.2}, t_{2.3}$ and $t_{2.4}$ and the buffers at their child sub-trees ({\rmi} level 3) can fit in the cache (highlighted with dotted lines). Therefore, we stop assigning buffers to the sub-trees at higher levels (i.e., {\rmi} level 4)}.

\noindent\textbf{Buffers Cache Complexity.} Assume $B$ is the cache line size in bytes, and $C$ is the cache capacity in bytes. According to~\cite{HL06}, the amortized cache complexity of one probe using the buffering scheme of a hierarchical index $X$ (e.g., trees or {\rmi}) should be less than or equal to $O(\frac{1}{B} \log_C n)$, where $n$ is the number of index nodes ({\rmi} models in our case). Otherwise, this buffering scheme is considered inefficient. Here, we show that our request buffers match the analytical bound of probing cache complexity, and hence proved to be efficient. \\

\begin{proposition}
\textit{The amortized cache complexity of one probe using our request buffers is $O(\frac{1}{B} \log_C n)$, where $n$ is the number of {\rmi} models, and hence matches the analytical bound from~\cite{HL06}.}
\label{prop:request_buffer_cache_complexity}
\end{proposition}

\begin{proof}
As mentioned before, we apply the request buffers on the {\rmi} index, starting from the first level, until a sub-tree $t$ and its buffers (excluding the buffer of this sub-tree's root) totally fit into the cache. Assume the number of models in this sub-tree $t$ is $k$. Then, using Proposition~\ref{prop:request_buffer_size}, the total size of the sub-tree $t$ and its buffers, excluding the buffer at $t$'s root, can be represented using this inequality $w(k \log_m k) < C$. By taking $\log_m$ for both sides, $(\log_m k + \log_m \log_m k) < \log_m \frac{C}{w}$. After eliminating lower-order terms, $\log_m k < \log_m \frac{C}{w}$.

Assume the size of probe key in bytes is $q$. Since we are using fixed-size buffers, the amortized cost of each probe on a sub-tree that fits in the cache is $O(\frac{1}{B})$, because the cost of loading sub-tree into the cache is amortized among $\frac{w(k \log_m k)}{q}$ probes. In addition, the number of sub-trees fitting in the cache that each probe should go through is $O(\frac{\log_m n}{\log_m k}) = O(\frac{\log_m n}{\log_m C}) = O(\log_C n)$. As a result, the amortized cost of each probe on the {\rmi} index is $O(\frac{1}{B} \log_C n)$, and meets the analytical bound in~\cite{HL06}. \\

\end{proof}

\subsection{Learned-based Hash {\inlj}}
\label{sec:linlj_rmi_hash}

\textcolor{black}{Instead of directly using a learned index (e.g., {\rmi} or {\rs}~\cite{KMR+20}) in {\inlj} (e.g., {\linlj}), we can employ a variant of a hash index that only uses the {\cdf}-based models, as a \textit{hash function}, to build and probe a typical hash table~\cite{SVH+21} (i.e., just using the {\cdf}-based models without storing keys in {\rmi}-specific data structures, such as gapped arrays, or applying any buffering optimization)}. Basically, if the input keys are generated from a certain distribution, then the {\cdf} of this distribution maps the data uniformly in the range $[0, 1]$, and the {\cdf} will behave as an order-preserving hash function in a hash table. In some cases, the {\cdf}-based models can over-fit to the original data distribution, and hence result in less number of collisions compared to traditional hashing, which would still have around 36.7\% regardless of the data distribution (based on the birthday paradox). For example, if the data is generated in an incremental time series fashion, or from sequential data with holes (e.g., auto generated IDs with some deletions), the predictability of the gaps between the generated elements determines the amount of collisions~\cite{SVH+21}, and can be captured accurately using learned models. Note that we build the {\cdf}-based models using the entire input relation, and hence its accuracy becomes high. Our experimental results (Section~\ref{sec:evaluation}) show the effectiveness of the learned-based hash variant in some real and synthetic datasets.



\section{Learned Sort-based Join}
\label{sec:lsj}

Sort-based join ({\sj}) is an essential join operation for any database engine. Although its performance is dominated by other join algorithms, like partitioned hash join~\cite{TAB+13}, it is still beneficial for many scenarios. For example, if the input relations are already sorted, then {\sj} becomes the best join candidate as it will skip the sorting phase, and directly perform a fast merge-join only. Another scenario is applying grouping or aggregation operations (e.g.,~\cite{CS94, AAD+96}) on join outputs. In such situation, the query optimizer typically picks a {\sj} as it will generate nearly-sorted inputs for the grouping and aggregate operations. In this section, we propose an efficient learned sort-based join algorithm, namely {\lsj}, which outperforms the existing state-of-the-art {\sj} algorithms.

\subsection{Algorithm Design}
\label{sec:lsj_overview}

One straightforward design to employ the learned models in {\sj} is to directly use a LearnedSort algorithm (e.g.,~\cite{KVC+20}) as a black-box replacement for the sorting algorithms used in the state-of-the-art {\sj} techniques, such as {\mpsm}~\cite{AKN12} and {\mway}~\cite{TAB+13}. This is an intuitive idea specially that the sorting phase is the main bottleneck in {\sj}, and using the LearnedSort is expected to improve the sorting performance over typical sorting algorithms as shown in~\cite{KVC+20}. However, based on our experiments, such design leads to a sub-par performance because LearnedSort is typically applied on different input partitions in parallel, and some of its processing steps (such as sampling, model building, and {\cdf}-based partitioning) can save redundant work, if they are designed to be shared among different partitions and subsequent {\sj} phases. 

Based on this insight, we design another {\sj} variant, namely {\lsj}, that supports work sharing, and hence reduces the overall algorithm latency. {\lsj} achieves that through building a {\cdf}-based model that can be reused to improve the different {\sj} phases (e.g., sort and join) by just scaling up its prediction degree. In addition to that, {\lsj} implements the good learned lessons from the {\sj} literature. In particular, {\mway}~\cite{TAB+13} shows that using a sorting algorithm with good ability to scale up on the modern hardware, like AVX sort~\cite{BAT+13}, is more desirable in a multi-core setting, even if it has a bit higher algorithmic complexity. Moreover, {\mpsm}~\cite{AKN12} concludes that getting rid of the merge phase increases the {\sj} throughput as merging is an intensive non-local write operation. {\lsj} leverages these two hints along with the {\cdf}-based model in the implementation of its phases as described in the rest of this section.

\subsection{{\lsj} Partition Phase}
\label{sec:lsj_partition_phase}

\begin{figure}
    \begin{center}
    \includegraphics[width=3.2in]{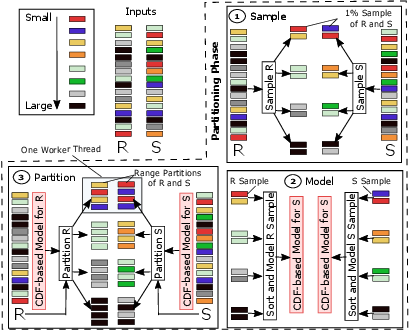}
    \end{center}
    \caption{\textcolor{black}{{\lsj} partitioning phase (the color of a key reflects how large it is).}}
    \label{fig:lsj_algorithm_1}
\end{figure}

The objective of this phase is to perform a range partitioning for both $R$ and $S$ relations among the different worker threads, such that the output partitions of each relation become relatively sorted and allow the {\lsj} subsequent phases to work totally in parallel (i.e., without any synchronization). 

\noindent\textbf{Key Idea.} The partitioning idea is straightforward. We build a {\cdf}-based model for each input relation based on a small data sample, and then use this model (instead of the traditional radix partitioning~\cite{TAB+13}) to range-partition the relation among the different worker threads. \textcolor{black}{Note that we follow the partitioning theme of {\mway}~\cite{TAB+13} where we range-partition both relations, not just the smaller one as in {\mpsm}~\cite{AKN12}. Although this makes the partitioning phase has a bit worse performance, the extra partitioning overhead is not a major bottleneck and already dominated by its effect on improving the sorting phase as described later (check Section~\ref{sec:lsj_sort_phase}).}


\noindent\textbf{Model Building.} \textcolor{black}{As shown in the first two steps of Figure~\ref{fig:lsj_algorithm_1}, we build the model for each relation as follows: (1)~in parallel, each worker samples around 1\% of the input tuples, (2)~the sample in each worker is then sorted, and after that, all sorted samples from the different workers are merged into one worker to train and build the model.}

\noindent\textbf{Range Partitioning.} \textcolor{black}{Inspired by~\cite{AKN12, TAB+13}, we implement this operation (the third step in Figure~\ref{fig:lsj_algorithm_1}) as follows}: (1)~each worker locally partitions its input chunk using the built model. While partitioning, each worker keeps track of the partitions' sizes and builds a local histogram for the partitions at the same time, then (2)~the local histograms are combined to obtain a set of prefix sums where each prefix sum represents the start positions of each worker’s partitions within the target workers, and finally (3)~each worker \textit{sequentially} writes its own local partitions to the target workers. Note that there is no separate step for building the local histograms, as we directly use the model to partition the data into an over-allocated array to account for skewed partitions. This approach is adapted from a previous work on the LearnedSort~\cite{KVC+20} and proves its efficiency. Typically, with a good built model, the partitions are balanced and the over-allocation overhead is 2\%.

\noindent\textbf{Software Write-Combine ({\swwc}) Buffers.} To further reduce the overhead of local partitioning, we used a well-known {\swwc} optimization, combined with non-temporal streaming~\cite{TAB+13, SCD16}. Basically, {\swwc} reduces the TLB pressure by keeping a cache-resident buffer of entries for each partition, such that these buffers are filled first. When a buffer becomes full, its entries are flushed to the final partition. With non-temporal streaming, these buffers bypass all caches and are written directly to DRAM to avoid polluting the caches with data that will never be read again.


\subsection{{\lsj} Sort Phase}
\label{sec:lsj_sort_phase}

\begin{figure}
    \begin{center}
    \includegraphics[width=3.2in]{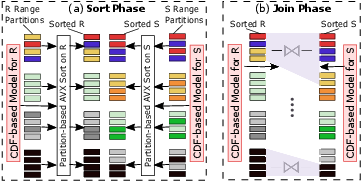}
    \end{center}
    \caption{\textcolor{black}{{\lsj} sorting and joining phases (Continued example from Figure~\ref{fig:lsj_algorithm_1}).}}
    \label{fig:lsj_algorithm_2}
\end{figure}

\begin{figure}
    \begin{center}
    \includegraphics[width=3.2in]{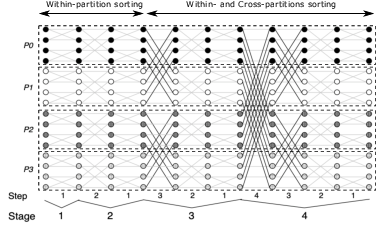}
    \end{center}
    \caption{Example of a bitonic sorting network for 16 items with 4 partitions.}
    \label{fig:lsj_bitonic}
\end{figure}

Once the partitioning phase generates the output $R$ and $S$ range partitions for the different workers, each worker sorts its own partitions locally. At the end of the sorting phase, both relations will be globally sorted and ready to be directly joined as described in Section~\ref{sec:lsj_join_phase}. Typically, the sorting phase is implemented in a multi-threaded environment using an {\simd} variation of the bitonic sort, namely AVX sort~\cite{BAT+13}. Although the complexity of bitonic sort is $O (n \log^2 n)$, which is a bit higher than the standard $O(n \log n)$ complexity, the bitonic sort is better for parallel implementation on modern hardware as it compares elements using min/max instructions in a predefined sequence that never depends on the data. Here, we propose an new efficient variation of the AVX sort that exploits {\cdf}-based partitioning to reduce its number of comparisons. We also discuss its complexity and trade-offs.

\noindent\textbf{Key Idea.} \textcolor{black}{On each worker, the proposed AVX sort variation takes the local range partition (i.e., output of the partitioning phase (Section~\ref{sec:lsj_partition_phase})) of one relation, say $R$, along with the {\cdf}-based model of this relation as inputs (Figure~\ref{fig:lsj_algorithm_2}(a)). It first uses the model to further divide the local partition into non-overlapping partitions $p$ that are relatively sorted, then apply the AVX sort on each partition independently.} The number of partitions $p$ is tuned to provide a good balance between parallelism and caching performance. Note that we reuse the same {\cdf}-based model from the partitioning phase, but we scale its prediction up for higher degree of partitioning.

\noindent\textbf{Complexity Analysis.} To compute the complexity of the new AVX sort, let $n$ be the number of keys to be sorted, and $p$ be the number of {\cdf}-based partitions. Each partition has $\frac{n}{p}$ keys roughly. Figure~\ref{fig:lsj_bitonic} shows an example of a typical bitonic sorting network (the core of AVX sort) for 16 keys ($n = 16$) and 4 partitions ($p = 4$). It also highlights the type of each sorting comparison, whether it occurs within the same partition (grey arcs) or across different partitions (black arcs). As shown, the first $\log \frac{n}{p}$ stages execute completely local (within the same partition). For a subsequent stage $\log \frac{n}{p} + k$ where $ 1 \leq k \leq \log p$, the first $k$ steps compare keys from different partitions, while the last $\log \frac{n}{p}$ steps are completely local. As mentioned before, the total complexity of typical bitonic sort is $O (n \log^2 n)$, which is basically the total number of comparisons whether they are within the same partitions or not. However, using our proposed variation, we first need $O(n)$ for the {\cdf}-based partitioning, and then a total of $O(\frac{1}{2} n (\log \frac{n}{p}) (\log \frac{n}{p} + 1))$ for the comparisons occurring in the first $\log \frac{n}{p}$ stages (within the same partition only). Note that our algorithm will skip the remaining stages as the partitions are already relatively sorted. So, the total complexity of our algorithm is $O(n + \frac{1}{2} n (\log \frac{n}{p}) (\log \frac{n}{p} + 1)) = O(n \log^2 \frac{n}{p})$.  

\noindent\textbf{Partitioning Fan-out.} As shown from the complexity analysis, the number of partitions $p$ plays a key role in improving the AVX sort throughput. Increasing $p$ leads to creating more small partitions that need less number of AVX sort stages to be sorted. In addition, these partitions will fit nicely in the cache, and hence improves the temporal locality during sorting as well. On the other hand, having a large number of partitions may lead to a different type of cache problem. It may cause TLB misses if the number of created partitions is too large. We found that estimating the value of $p$ to be around 60\% of the number of L1/L2 TLB entries provides a good balance between TLB misses and cache fitting.


\subsection{{\lsj} Join Phase}
\label{sec:lsj_join_phase}

Since the two relations $R$ and $S$ become globally sorted, we skip any merge phase and directly merge-join them.

\noindent\textbf{Key Idea.}  Assuming that $S$ is the smaller relation, the basic idea is to perform a join between each partition from $R$ and all its \textit{overlapping} partitions from $S$ (Note that these partitions are the final ones generated in the sorting phase, right before applying the AVX sort). 

\noindent\textbf{Chunked-Join.} The join operation is performed in two steps: 1)~for any partition in $R$, we identify the start and the end positions of its overlapping partitions in $S$ \textcolor{black}{(e.g., the first partition from $R$ only overlaps with the first two partitions from $S$ in Figure~\ref{fig:lsj_algorithm_2}(b))}. To do that, we query the {\cdf}-based model of $S$ - scaled up by the same factor used in the sorting phase of $S$ to generate its partitions - with the values of the first and the last keys of this partition from $R$ to return the corresponding first and last overlapping partitions from $S$, namely $f_S$ and $l_S$, respectively. Then, 2)~we merge-join any partition from $R$ with each partition in $S$ belongs to the corresponding overlapping range $[f_S, l_S]$ (including $f_S$ and $l_S$) to find the final matching tuples. Note that all joins from different workers are done in parallel.

We reuse the same fine-grained partitions from the sorting phase in the join operation for two main reasons. First, this ensures using the models correctly when predicting the overlapping partitions as we reuse the same {\cdf} scale up factor when querying the same model. Second, such partitions are probably small enough to fit in the cache (their sizes are already tuned for the AVX sort using the partitioning fan-out).

\subsection{{\lsj} Complexity Analysis}
\label{sec:lsj_complexity_analysis}
Assume $W$ to be the number of worker threads used to run~{\lsj}. Let $N_R$ and $N_S$ to be the number of tuples in relations $R$ and $S$, respectively. Let $S_R$ and $S_S$ to be the number of sampled tuples from relations $R$ and $S$, respectively, to build their {\cdf}-based models. Also, assume $P_R$ and $P_S$ to be the total number of final sorted partitions in relations $R$ and $S$, respectively, where $P_R \geq W$ and $P_S \geq W$. Assume $O_{R}$ to be the maximum number of workers that have overlapping partitions in $R$ with any tuple in $S$ ($0 \leq O_R \leq W$). Similarly, assume $O_{S}$ to be the maximum number of workers that have overlapping partitions in $S$ with any tuple in $R$ ($0 \leq O_S \leq W$). We can roughly compute the complexity of {\lsj} for each worker thread to be the total of the following complexities, corresponding to its different phases:

\begin{itemize}
    \item $[\frac{S_R + S_S}{W} + S_R \log S_R + S_S \log S_S]$ (for sampling and building the {\cdf}-based models for $R$ and $S$)
    \item $[\frac{N_R + N_S}{W}]$ (for range partitioning two chunks from $R$ and $S$ with sizes $\frac{N_R}{W}$ and $\frac{N_R}{S}$, respectively)
    \item $[\frac{N_R}{W} \log^2 \frac{N_R}{P_R} + \frac{N_S}{W} \log^2 \frac{N_S}{P_S}]$ (for sorting two chunks from $R$ and $S$ with sizes $\frac{N_R}{W}$ and $\frac{N_S}{W}$, respectively)
    \item $[\frac{N_R}{W} O_S + \frac{N_S}{W} O_R]$ (for processing two chunks from $R$ and $S$ with sizes $\frac{N_R}{W}$ and $\frac{N_S}{W}$, respectively, in the join phase)
\end{itemize}





\section{Learned Hash-based Join}
\label{sec:lhj}

Here, we discuss our trials to improve the performance of hash-based joins ({\hj}) using {\rmi} and {\cdf}-based partitioning (we refer to them here as learned models). Similar to learned-based hash {\inlj} (Section~\ref{sec:linlj_rmi_hash}), the main idea is to replace the hash functions used in both non-partitioned ({\npj}) and partitioned ({\pj}) hash joins with learned models. Unfortunately, learned models showed a poor performance (i.e., high overall latency) in both {\npj} and {\pj} algorithms. This is mainly because of the poor accuracy of the model used. \textcolor{black}{In the case of learned-based hash {\inlj}, a pre-built learned index is assumed to exist on one of the input relations (i.e., indexing time is not included in the join time). Therefore, the learned models in such index can be built offline from the entire dataset to accurately capture the underlying data distribution, and have a very competitive performance with hashing}. However, in the {\hj} algorithms, \textcolor{black}{the input relations are not pre-indexed (e.g., intermediate outputs of two operators). Therefore, the learned models need to be built \textit{during the join execution}, which is infeasible using the entire input relations as it will significantly increase the join latency}. Therefore, we build the learned models with a small data sample (typically ranges from 1\% to 10\% of the input relations). This makes the model accuracy poor, and leads to inefficient hash tables building and probing in both ({\hj}) algorithms (Check Section~\ref{sec:evaluation} for more evaluation results).

\section{Experimental Evaluation}
\label{sec:evaluation}

In this section, we experimentally study the performance of the learned and non-learned variations of indexed nested loop joins ({\inlj}), sort-based joins ({\sj}) and hash-based joins ({\hj}). Section~\ref{sec:experimental_setup} describes the competitors, datasets, hardware and metrics used. Section~\ref{sec:evaluation_results_all} analyzes the performance under different dataset types, sizes and characteristics. It also analyzes the joins under different threading settings and via performance counters (e.g., cache misses). Sections~\ref{sec:evaluation_results_inlj_only} and ~\ref{sec:evaluation_results_sj_only} analyze the performance while tuning specific parameters for the learned {\inlj} and {\sj} algorithms, respectively.

\subsection{Experimental Setup}
\label{sec:experimental_setup}


\noindent\textbf{Competitors.} For {\inlj}, we evaluate two categories of indexes.

\underline{\textit{Learned Indexes}.} We provide the following three variants of {\rmi}: (1)~typical {\rmi} ({\rmiinlj}) as described in Section~\ref{sec:background_rmi_cdf_partitioning}, and both (2)~gapped ({\linlj}) and (3)~buffered gapped {\rmi} ({\blinlj}) as described in Section~\ref{sec:linlj}. We also provide a variant of a learned-based hash index ({\rshashinlj}), that uses {\rs}~\cite{KMR+20} as a hash function (described in Section~\ref{sec:linlj_rmi_hash}). We selected {\rs} instead of {\rmi} in this variant to explore more learned models in our experiments. Finally, we include a variant of the updatable learned index ({\alexinlj}), that uses {\alex}~\cite{DMY+20}, which is very relevant to our proposed {\grmi}. The implementations of {\rmi}, {\rs} and {\alex} are obtained from their open-source repositories~\cite{RS20, RMI20, ALEX20}. The {\rmi} hyper-parameters are tuned using CDFShop~\cite{MZK20}, an automatic {\rmi} optimizer. {\rs} is manually tuned by varying the error tolerance of the underlying models. {\alex} is configured according to the guidelines in its published paper~\cite{DMY+20}.

\underline{\textit{Non-Learned Indexes}.} We also report the {\inlj} performance with three tree-structured indexes (the Adaptive Radix Trie~\cite{LKN13} ({\arttreeinlj}), the Cache-Sensitive Search Trees~\cite{HYF+08} ({\csstreeinlj}), and a buffered variant of the Cache-Sensitive Search Tree ({\bcsstreeinlj}), which utilizes the proposed hierarchical buffering in Section~\ref{sec:linlj_optimized_buffer}\footnote{The buffer is defined per an index tree node instead of an {\rmi} model.}), and two hash indexes (a bucket chaining hash index ({\hashinlj}), and a cuckoo hash index ({\cuckooinlj})). We build each index using the entire input relation (i.e., inserting every key). In addition, the implementations of tree-structured indexes are provided by their original authors, while the cuckoo hash index is adapted from its standard implementation that is used in~\cite{SOSD20}. 

For {\sj}, we compare our proposed {\lsj} (Section~\ref{sec:lsj}) with four sort-based join baselines, coming from the stat-of-the-art algorithms {\mway} and {\mpsm}. Unless otherwise mentioned, we use the optimized variant of our proposed algorithm, referred to as {\blsj}, which includes the {\swwc} buffering and non-temporal streaming optimizations. For {\mway}, we have two baselines: (1)~an adapted variation of the original open-source implementation in~\cite{MWAY20} that works with AVX 512, referred to as {\mwaysj} (tuned its parameters according to a recent experimental study~\cite{SCD16}), and (2)~a variation of {\mway}, referred to as {\mwaylssj}, that uses the LearnedSort algorithm~\cite{KVC+20} as a black-box in its sorting phase. Regarding {\mpsm}, we have two baselines: (1)~a best-effort implementation of the original algorithm~\cite{AKN12} on our own (unfortunately, the original code is not available for public), referred to as {\mpsmsj}, and (2)~a variation of {\mpsm}, referred to as {\mpsmlssj}, that uses the LearnedSort algorithm~\cite{KVC+20} as a black-box in its sorting phase.

For {\hj}, we use the reference implementations of non-partitioned and partitioned hash joins from~\cite{MWAY20}, referred to them as {\npjhj} and {\pjhj}, respectively, and adapted two learned variations from them by replacing their hash functions with learned models, referred to them as {\lnpj} and {\lpj}. We followed the guidelines in~\cite{SCD16} to optimize the performance of all {\hj} algorithms.

\noindent\textbf{Datasets.} We evaluate all join algorithms with real and synthetic datasets. For real datasets, we use three datasets from the {\sosd} benchmark~\cite{MKR+20}, where we perform a self-join with each dataset. Any dataset is a list of unsigned 64-bit integer keys. We generate random 8-byte payloads for each key. The real datasets are:
\begin{itemize}
    \item {\fb}: 200 million randomly sampled Facebook user IDs.
    \item {\osm}: 800 million cell IDs from Open Street Map. 
    \item {\wiki}: 200 million timestamps of edits from Wikipedia.
\end{itemize}

For synthetic datasets, unless otherwise stated, we generate three datasets with different data distributions for each input relation (either $R$ or $S$). Any dataset consists of a list of 32-bit integer keys, where each key is associated with 8-byte payload. Each dataset has five size variants (in number of tuples): 16, 32, 128 and 640, and 1664 million tuples. All of them are randomly shuffled. The synthetic datasets are:
\begin{itemize}
    \item {\seqh}: sequential IDs with 10\% random deletes (holes). 
    \item {\unif}: uniform distribution, with min = 0 and max = 1, multiplied by the size, and rounded to the nearest integer.
    \item {\lognorm}: lognormal distribution with $\mu = 0$ and $\sigma = 1$ that has an extreme skew (80\% of the keys are concentrated in 1\% of the key range), multiplied by the dataset size, and rounded to the nearest integer.
\end{itemize}

\noindent\textbf{Hardware and Implementation.} All experiments are conducted on an Arch Linux machine with 256 GB of RAM and an Intel(R) Xeon(R) Gold 6230 CPU @ 2.10GHz with Skylake micro architecture (SKX) and L3 cache of 55MiB. \textcolor{black}{The machine has 2 NUMA nodes, each has 40 CPUs}. All approaches are in C++, and compiled with GCC (11.1.0). 

\noindent\textbf{Metrics.}  We use the join throughput as the default metric. We define it as in~\cite{SCD16} to be ratio of the sum of the relation sizes (in number of tuples) and the total runtime $\frac{|R| + |S|}{|runtime|}$.

\noindent\textbf{Default Settings.} Unless otherwise stated, all reported numbers are produced with running 32 threads. For {\blinlj}, we use 4 gaps per key. To fairly compare {\blinlj} with hash-based indexes, we build the hash indexes to match the gaps setting of {\blinlj} by (1)~using a bucket size of 4 in the block chaining index ({\hashinlj}), and (2)~creating a 2-table Cuckoo hash index ({\cuckooinlj}), where the size of these tables are 4 times the size of the input relation (i.e., load factor 25\%). \textcolor{black}{For both {\blinlj} and {\bcsstreeinlj}, we keep assigning buffers (as mentioned in Section~\ref{sec:linlj_optimized_buffer}), starting from the first index level (root is level 0), till we find that a sub-tree (and its buffers), rooted at a model (or an index node in case of {\bcsstreeinlj}) in level $l$, can fit in the L3 cache. Then, we stop assigning buffers to sub-trees at levels higher than $l$}. For {\blsj}, we use a partitioning fan-out of 10000. For any {\lsj} variant, the number of range (i.e., non-overlapping) partitions is set to the number of running threads to ensure their work independence in the sorting and joining phases. For datasets, all duplicate keys are removed to make sure that {\arttreeinlj} properly works (a separate experiment for duplicates exists).


\subsection{{\inlj}, {\sj}, and {\hj} Evaluation Results}
\label{sec:evaluation_results_all}
\begin{figure*}
     \begin{center}
     \includegraphics[width=6.8in]{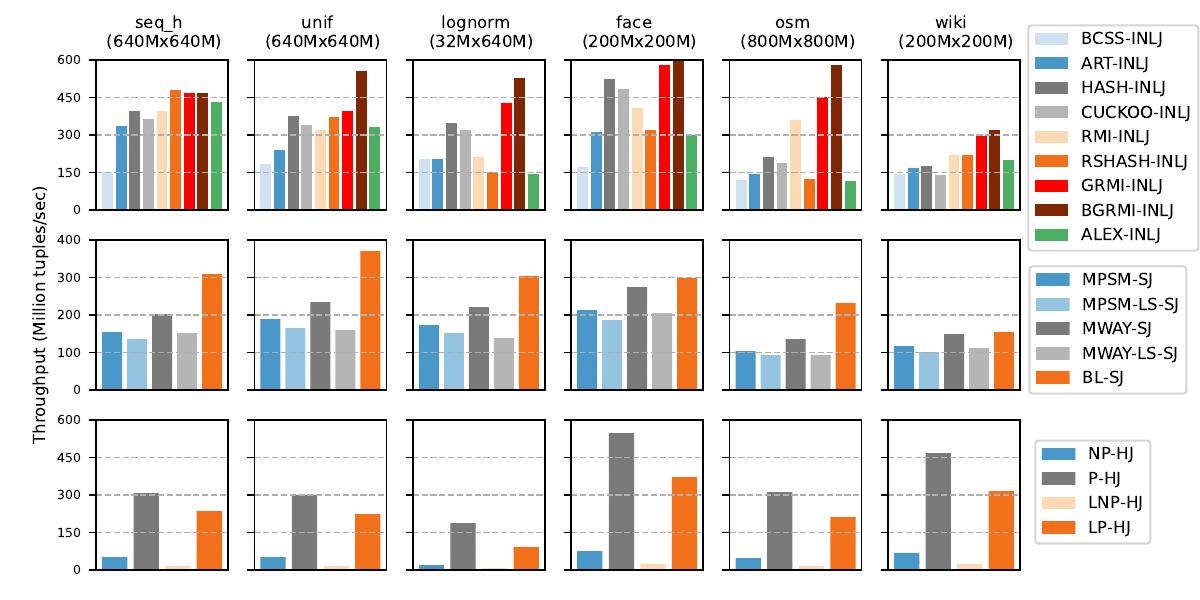}
     \end{center}
     \caption{Performance of the three join categories for real and synthetic datasets, where each row represents a category.}
     \label{fig:all_real_synthetic}
\end{figure*}
 
\noindent\textbf{Real and Synthetic Datasets.} Figure~\ref{fig:all_real_synthetic} shows the join throughput (in million tuples per sec) of the proposed learned joins against their traditional competitors in each join category. 

Clearly, {\blinlj} outperforms all other {\inlj} competitors in all real and synthetic datasets, except {\seqh}. This is attributed to the effect of building the gapped {\rmi} using "model-based" insertion and employing the buffering optimization during key lookups. Regardless the {\rmi} prediction accuracy for any dataset, each key is located very close to the position where it is expected to be found, and by using the exponential search, the key is typically found after 1 or 2 search steps (using binary search in {\rmiinlj} takes at least 4 search steps). {\blinlj} is better than {\linlj} in almost all cases with a throughput improvement between 5\% and 25\%. {\rshashinlj} shows a slightly better performance than {\blinlj} with the {\seqh} dataset due to the model over-fitting in the sequential cases (Check Section~\ref{sec:linlj_rmi_hash} for more details). In contrast, {\rshashinlj} has a very poor performance with {\lognorm} and {\osm} datasets. This is expected as these datasets are more skewed than others, and lead to building large sparse radix tables in {\rs}~\cite{KMR+20}. Therefore, the learned model becomes less effective in balancing keys that will be queried in the join operation. Similarly, {\alexinlj} suffers from the extra overhead of its B-tree-like structure. This matches the results reported by our online leaderboard~\cite{SOSDLeaderboard20} for benchmarking learned indexes: all updatable learned indexes are significantly slower. Unlike {\alex}, the proposed {\grmi} variants only injected {\alex}'s ideas of model-based insertion and exponential search into {\rmi} (without internal nodes). We also found two more interesting observations. First, {\hashinlj} is slightly better than {\cuckooinlj}. This is because, in {\cuckooinlj}, each key has at most 2 hash table lookups (in case of having collision), which incur 2 cache misses. In contrast, the checks occurring in bucket chaining (with bucket size of 4) are mostly done on keys in the same cache line (1 cache miss). Such observation confirms the reported findings in a recent hash benchmarking study~\cite{RAD15}. \textcolor{black}{Second, the proposed hierarchical buffering in Section~\ref{sec:linlj_optimized_buffer} improves the performance of {\inlj} with non-learned tree-structured indexes as well. In particular, {\bcsstreeinlj} outperforms {\csstreeinlj} in all datasets with at least 9\% (when using {\wiki} datasets) and up to 15\% (when using {\unif} datasets) higher join throughput. Since {\csstreeinlj} is always the worst baseline, we removed it from all the reported results to avoid cluttered figures.} 

For {\sj}, {\blsj} consistently provides the highest join throughput in all datasets due to its effective sorting and join phases. In most datasets, {\blsj} has at least 50\% and 30\% less latency in the sorting and joining phases, respectively, compared to the {\mpsmsj} and {\mwaysj} baselines (More breakdown results are shown in Figure~\ref{fig:all_breakdown}). \textcolor{black}{In addition, just replacing the sorting phase, in both {\mway} and {\mpsm}, with the LearnedSort~\cite{KVC+20} algorithm actually reduces the throughput. This is mainly because, when using the LearnedSort in {\sj} algorithms, its overhead becomes very significant to hide with the improvement in the sorting runtime. Assuming $W$ is the number of workers, the LearnedSort is applied $W^2$ times in {\mwaylssj} because each worker sorts $W$ local partitions before applying the multi-way merge step (i.e., sorting $W^2$ partitions in total from all workers). Similarly, in {\mpsmlssj}, the overhead still exists, yet lighter (LearnedSort is repeated $W$ times only). Applying the LearnedSort many times amplifies the overhead of allocating/de-allocating its internal data structures (e.g., over-allocated and spill buckets), and hence degrades the overall throughput. In contrast to {\mwaylssj} and {\mpsmlssj}, {\blsj} builds one CDF-based model for each input relation and reuses it through the partitioning, sorting and joining phases, which saves any redundant work and increases the join throughput.}


For {\hj}, as expected, the learned variants (built with 2\% samples) demonstrate very low throughputs especially for the real datasets as they are more challenging. They show at least 15X slowdown in most of the cases. Thus, we exclude the {\hj} results from all figures in the rest of this section. 

\begin{figure}
    \begin{center}
     \includegraphics[width=3.2in]{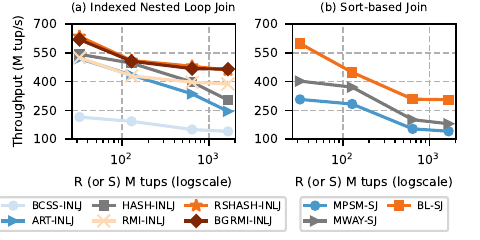}
     \end{center}
     \caption{Performance when changing number of tuples in the {\seqh} dataset.}
     \label{fig:all_dataset_sizes}
\end{figure}
\noindent\textbf{Dataset Sizes.} Figure~\ref{fig:all_dataset_sizes} shows the performance of different {\inlj} and {\sj} algorithms while scaling up the number of tuples to be joined. In this figure, we report the throughput for joining four variants of the {\seqh} dataset: 32Mx32M, 128Mx128M, 640Mx640M, and 1664Mx1664M (Note that the figure has a logscale on its x-axis). 

For {\inlj}, all learned variants are able to smoothly scale to larger dataset sizes, with a small performance degradation. In case of {\rmiinlj}, the "last-mile" search step will suffer from algorithmic slowdown only as it applies binary search, and such slowdown becomes even less in the {\blinlj} as it employs an exponential search instead. Moreover, both tree structures, {\arttreeinlj} and {\bcsstreeinlj}, become significantly ineffective at very large datasets because traversing the tree becomes more expensive compared to performing either binary or exponential search in the learned variants. For {\sj}, all algorithms keep the same relative performance gains regardless the sizes of input datasets. 

\begin{figure}
    \begin{center}
     \includegraphics[width=3.2in]{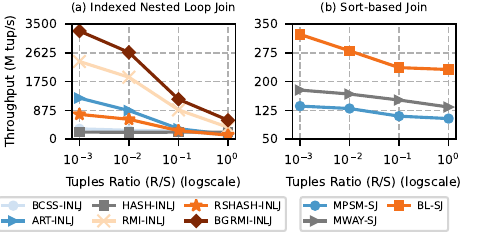}
     \end{center}
     \caption{\textcolor{black}{Performance when changing the datasets ratio (the {\osm} dataset)}.}
     \label{fig:all_datasets_ratio}
\end{figure}

\noindent\textcolor{black}{\textbf{Datasets Ratio.} Figure~\ref{fig:all_datasets_ratio} shows the performance of different {\inlj} and {\sj} algorithms while joining input datasets of different sizes. In this experiment, we fix the size of one input relation $S$, which is the {\osm} dataset with 800M tuples, while varying the number of tuples of the second input relation $R$ as a ratio from $S$. We generate four variants of the relation $R$ with the following number of tuples: 0.8M (0.1\% of $S$), 8M (1\% of $S$), 80M (10\% of $S$), and 800M (100\% of $S$). Note that the figure has a logscale on its x-axis. In case of the {\inlj} algorithms, we index relation $R$ to study the effect of changing the index size on the join performance.}

\textcolor{black}{Clearly, the throughputs of {\blinlj} and {\rmiinlj} at small datasets ratio, including 0.001 and 0.01, are significantly better than the throughputs of other algorithms, specially {\hashinlj} and {\bcsstreeinlj} (at least 2X better). The main reason for that is the extremely small size of learned indexes at these ratios (compared to the tree structures and hash tables), which can completely fit into the cache (i.e., almost no cache misses). In general, the performance gap between all {\inlj} techniques significantly decreases at larger dataset ratios, including 0.1 and 1, where the rate of cache misses during index lookups becomes higher. We can also observe that both learned and non-learned {\sj} algorithms follow a similar performance trend as in {\inlj} algorithms, yet, with lower throughputs. This is expected as {\sj} includes more overhead coming from the partitioning, sorting, and merging phases (i.e., not just lookups during the join operation as in {\inlj}).}

\begin{figure}
    \begin{center}
     \includegraphics[width=3.2in]{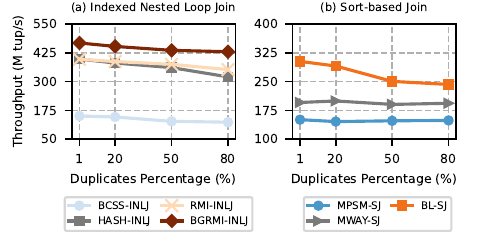}
     \end{center}
     \caption{Performance when changing number of duplicates in the {\seqh} dataset.}
     \label{fig:all_duplicates}
\end{figure}

\noindent\textbf{Duplicates.} Here, we study the performance of the {\inlj} and {\sj} algorithms while changing the ratio of duplicate keys in input datasets (i.e., covering 1-to-$m$ and $n$-to-$m$ joins). In this experiment, given a percentage of duplicates $x$\%, we replace $x$\% of the unique keys in each 640M {\seqh} dataset with repeated keys. Our {\lsj} naturally supports joining duplicates through the Chunked-Join operation (Section~\ref{sec:lsj_join_phase}). To support joining duplicates in {\blinlj}, we perform a range scan between the error boundaries provided by {\rmi} around the predicted location, similar to~\cite{MKR+20}.  

As shown in Figure~\ref{fig:all_duplicates}, the number of duplicates affects the join throughput of our learned algorithms as the built models will always predict the same position/partition to the key, which increases the number of collisions.
However, the performance of learned algorithms degrade gracefully and still remains better than non-learned algorithms. 
Note that the throughputs of {\bcsstreeinlj} and {\hashinlj} are decreased as well due to the increase in their tree depth and bucket chain, respectively, while {\mpsm} and {\mway} are almost not affected.

\begin{figure}
    \begin{center}
     \includegraphics[width=3.2in]{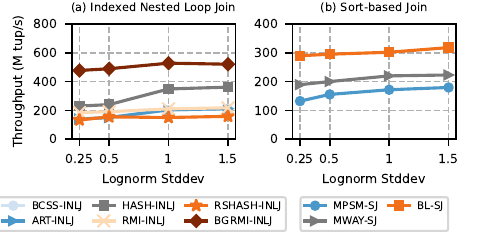}
     \end{center}
     \caption{\textcolor{black}{Performance when changing the skewness in the {\lognorm} dataset}.}
     \label{fig:all_skewness}
\end{figure}

\noindent\textcolor{black}{\textbf{Skewness.} In this experiment, we study the robustness of different {\inlj} and {\sj} algorithms while changing the skewness degree of input datasets. Figure~\ref{fig:all_skewness} shows the throughput for joining {\lognorm} datasets, where one input relation is fixed and the other one is varied according to the skewness degree (i.e., {$\sigma$}). As a fixed relation, we use the {\lognorm} dataset with 32M tuples that was generated using the default parameters mentioned in Section~\ref{sec:experimental_setup}. For the varied relation, we generate four new {\lognorm} datasets, each contains 640M tuples, yet, with different $\sigma$ values: 0.25, 0.5, 1 and 1.5. We index the varied relation when using the {\inlj} algorithms.}    

\textcolor{black}{As shown in the figure, all algorithms show pretty similar performance at $\sigma$ values of 1 and 1.5, except {\blsj} which has a bit better throughput. However, the difference in performance becomes significant at $\sigma$ values of 0.5 and 0.25 (i.e., extremely skewed datasets). In case of the {\inlj} algorithms, all learned variants are less susceptible to changing the skewness degree, except {\rshashinlj}. This is due to the ability of {\cdf}-based learned models to balance keys over the index array, and avoid generating extremely large clusters of keys. As expected, {\rshashinlj} has less join throughput because increasing the skewness degree (i.e., decreasing the $\sigma$ value) will increase the sparsity of the generated radix tables in {\rs}~\cite{KMR+20}, and hence leads to large clusters of keys. On the other side, the throughputs of {\bcsstreeinlj} and {\hashinlj} are significantly decreased by increasing the skewness degree due to the increase in their tree depth and bucket chains, respectively. In case of the {\sj} algorithms, {\mpsmsj} is the most affected algorithm by increasing the skewness degree because its range partitioning technique can easily result in unbalanced partitions among different worker threads (i.e., some threads will perform much work than others), which increases the overall algorithm latency. In contrast, {\blsj} and {\mwaysj} are almost not affected. This is because {\blsj} employs a {\cdf}-based learned model to perform the range partitioning, and hence the work is balanced among different threads. In addition, {\mwaysj} implements a specific strategy to handle extremely large tasks (i.e., avoiding heavy hitters), which decreases the effect of skewness~\cite{BAT+13}.}

\begin{figure}
    \begin{center}
     \includegraphics[width=3.2in, height=1.6in]{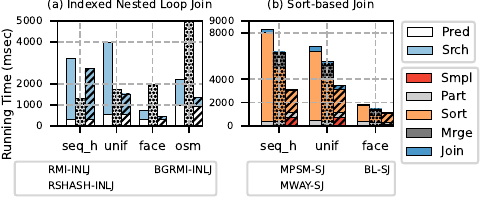}
     \end{center}
     \caption{Runtime breakdown for the phases of learned {\inlj} and {\sj} algorithms.}
     \label{fig:all_breakdown}
\end{figure}

\noindent\textbf{Performance Breakdown.} Figure~\ref{fig:all_breakdown} shows the runtime breakdown (in msec) for the different {\inlj} and {\sj} algorithms when joining different datasets. For synthetic datasets ({\seqh} and {\unif}), we show the results of joining the 640Mx640M variant.
For the learned {\inlj} case, we breakdown the runtime into {\rmi} prediction (\textit{Pred}) and final "last-mile" search (\textit{Srch}) phases. Since {\rshashinlj} does not apply a final search step (it just uses the {\rs} prediction as a hash value), we consider its total time as prediction (i.e., "0" search phase). For all datasets, except {\seqh}, the search phase of {\blinlj} is at least 3 times faster than typical {\rmiinlj}, while the prediction time is almost the same.   

For the {\sj} case, we breakdown the runtime into sampling (\textit{Smpl}), partitioning (\textit{Part}), sorting (\textit{Sort}), merging (\textit{Mrge}), and finally, merge-joining (\textit{Join}). Clearly, the sorting phase is dominating in each algorithm, and improving it will significantly boost the {\sj} performance. On average, {\blsj} has a 3X and 2X faster sorting phase compared to {\mpsm} and {\mway}, respectively. Also, the {\blsj} joining phase is at least 1.7X faster than its counterparts in {\mpsm} and {\mway}.

\begin{figure}
    \begin{center}
     \includegraphics[width=3.2in]{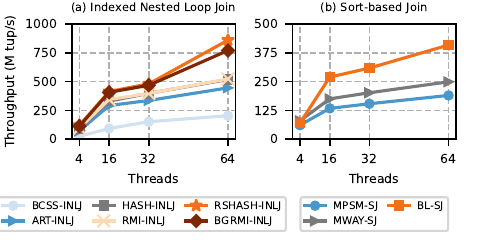}
     \end{center}
     \caption{Performance when changing number of threads.}
     \label{fig:all_parallelism}
\end{figure}
\noindent\textbf{Parallelism.} Figure~\ref{fig:all_parallelism} shows the performance of different {\inlj} and {\sj} algorithms while scaling up the number of threads from 4~to~64. In this experiment, we report the throughput for joining the 640Mx640M variant of the {\seqh} dataset. 

Overall, all {\inlj} and {\sj} algorithms scale well when increasing the number of threads, although only the learned join variants of {\inlj} achieve a slightly higher throughput than others. The main explanation for this good performance is that learned {\inlj} variants usually incur less cache misses (Check the performance counters results in Figure~\ref{fig:all_performance_counters}), because they use learned models that have small sizes and can nicely fit in the cache. Since threads will be latency bound waiting for access to RAM, then threading in typical {\inlj} algorithms will be more affected with latency than learned variants, and degrades its performance. This observation is confirmed by a recent benchmarking study for learned indexes~\cite{MKR+20}.



\begin{figure*}
     \begin{center}
     \includegraphics[width=6.8in]{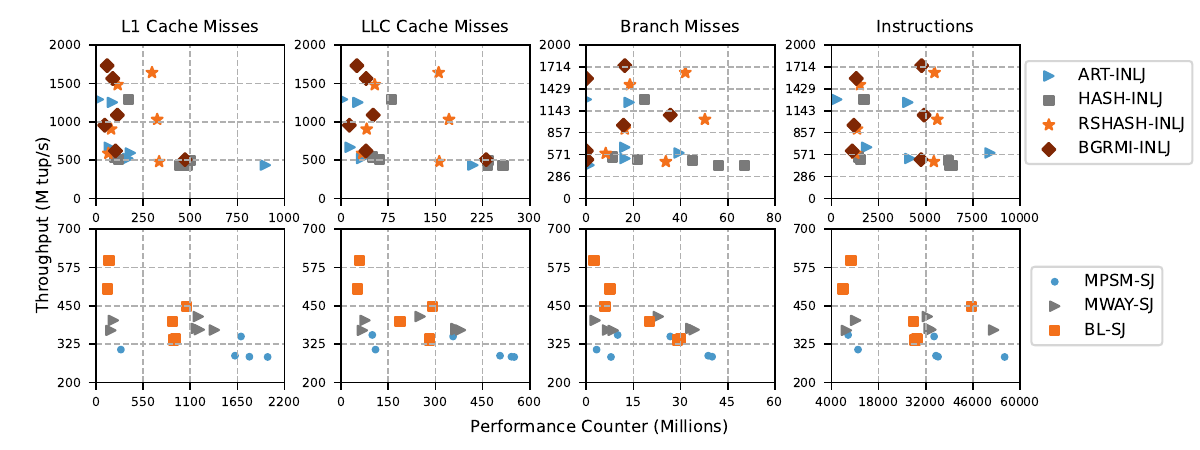}
     \end{center}
     \caption{Low-level performance counters for exploratory analysis, where each row represents a category of join. 
     }
     \label{fig:all_performance_counters}
\end{figure*}

\noindent\textbf{Other Performance Counters.}
To better understand the behavior of the proposed learned join algorithms, we dig deep into their low-level performance counters, and try to find any correlation between these counters and the join throughput. Figure~\ref{fig:all_performance_counters} presents the plotting of four performance counters (in millions) for different {\inlj} and {\sj} algorithms, where each column represents a counter type and each row represents a join category. 
For {\inlj}, we can see that most of the {\blinlj} cases with high join throughputs occur at low cache misses, whether L1 or LLC. This is explainable as most of such misses happen in the "last-mile" search phase (a typical 2-levels {\rmi} will have at most 2 cache misses in its prediction phase, hopefully only 1 miss if the root model is small enough), and hence improving it by reducing the number of search steps will definitely reduce the total cache misses. For other counters, it seems that there is no obvious correlation between their values and the join throughput. 

In contrast, for {\sj}, there is a clear correlation between the {\blsj}'s counter values and the corresponding join throughput. The {\blsj} throughput increases by decreasing the number of cache and branch misses. This is mainly because the AVX sort is applied on small partitions that have better caching properties. In addition, the number of comparisons in the AVX sort itself becomes lower (Check the sorting phase complexity in Section~\ref{sec:lsj_sort_phase}), which in turn, reduces the number of branch misses as well.

\subsection{{\inlj}-Specific Evaluation Results}
\label{sec:evaluation_results_inlj_only}




\begin{figure}
     \begin{center}
     \includegraphics[width=3.2in]{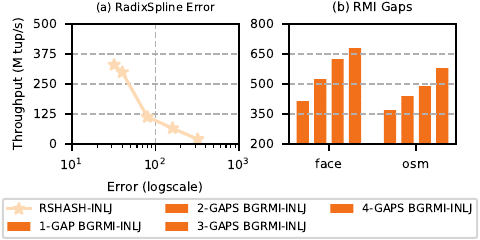}
    \end{center}
    \caption{Effect of changing {\rs} error and {\rmi} gap settings on the performance of {\rshashinlj} and {\blinlj}, respectively.}
     \label{fig:inlj_rs_error_rmi_gaps}
\end{figure}

\noindent\textbf{{\rs} Error.} Tuning the user-defined error value is an essential step for boosting the {\rs} performance~\cite{KMR+20}. Figure~\ref{fig:inlj_rs_error_rmi_gaps}(a) shows how the throughput of {\rshashinlj} is affected while changing the {\rs} error value (Note that the figure has a logscale on its x-axis). Here, we self-join the {\wiki} dataset. In general, increasing the error value will lead to more relaxed models with higher misprediction rate, which in turn build a hash table with higher collision rate. 

\noindent\textbf{{\rmi} Gaps.} Figure~\ref{fig:inlj_rs_error_rmi_gaps}(b) shows the effect of increasing the number of gaps assigned per key in {\blinlj}, while self-joining each of the {\fb} and {\osm} datasets. As expected, increasing the gaps degree allows for more space around each key to handle mispredictions efficiently, which in turn reduces the final search steps and increases the throughput. However, having more gaps will come with higher space complexity as well. Thus, this parameter needs to be tuned carefully.

\begin{figure*}
     \begin{center}
     \includegraphics[width=6.8in]{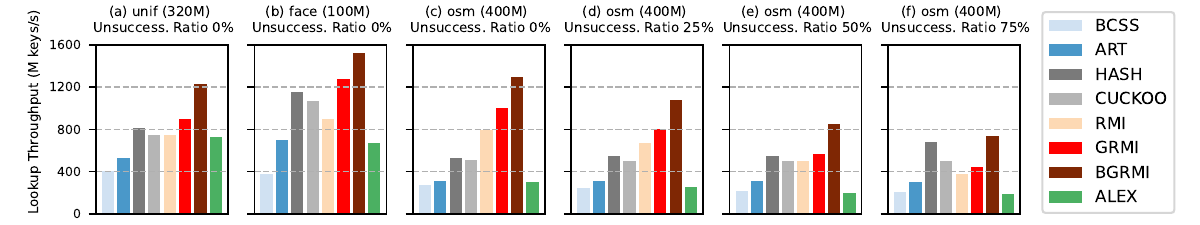}
     \end{center}
     \caption{\textcolor{black}{Performance of the {\grmi} and {\bgrmi} indexes compared to others, while changing the lookup workload sizes and unsuccessful ratios}.}
     \label{fig:all_real_synthetic_lookup_only}
\end{figure*}
\noindent\textcolor{black}{\textbf{Index Lookup.} In this experiment, we focus on studying the effectiveness of both {\grmi} and {\bgrmi} (compared to existing solutions) as general lookup indexes. Figure~\ref{fig:all_real_synthetic_lookup_only} shows the index lookup throughput in million keys per sec (y-axis) while changing the lookup workload. In the left three subfigures, we build indexes for the whole keys in the {\unif}, {\fb} and {\osm} datasets. Then, we generate a lookup workload from each dataset by randomly selecting, with replacement, 50\% of the indexed keys (i.e., all lookup keys already exist in the index, i.e., unsuccessful lookup ratio of 0\%). In the right three subfigures, we reuse the built {\osm} indexes before, yet, we change the {\osm} lookup workload by randomly replacing 25\%, 50\% and 75\% of its keys with non-existing ones to have three more {\osm} lookup workloads with unsuccessful ratios of 25\%, 50\% and 75\%, respectively. We build all indexes with the same default settings mentioned in Section~\ref{sec:experimental_setup}. Interestingly, the performance gap between the traditional {\rmi} and our proposed variants, i.e., {\grmi} and {\bgrmi}, decreases by increasing the the ratio of unsuccessful lookups ({\grmi} becomes even worse than {\rmi} at unsuccessful lookup ratio of 75\%). This is because both {\grmi} and {\bgrmi} perform the "last-mile" search using the exponential search (instead of the binary search used in {\rmi}), which is not the best solution when the search key is unlikely to be among the first neighbours around the model prediction. That being said, {\bgrmi} still has the best lookup performance in all cases as it combines both the model-based insertion idea and the hierarchical buffering optimization. We also observe that the advantage of HASH (i.e., bucket chaining) over CUCKOO becomes more significant at higher unsuccessful lookup ratios as well. This is due to the same reason mentioned before in Figure~\ref{fig:all_real_synthetic}: using CUCKOO (with 2 hash tables) guarantees 2 cache misses per each unsuccessful lookup, yet, HASH will probably do the required checks within the same cacheline (i.e., 1 cache miss). Finally, although BCSS is equipped with the hierarchical buffering optimization, its performance is still greatly penalized by the need to traverse the increased tree depths, specially in large datasets.}




\subsection{{\sj}-Specific Evaluation Results}
\label{sec:evaluation_results_sj_only}

\begin{figure}
     \begin{center}
     \includegraphics[width=3.2in]{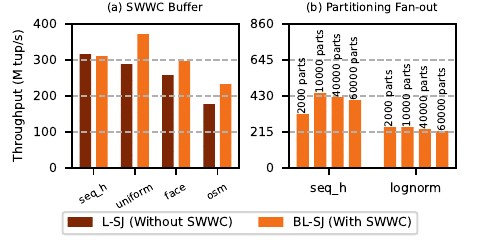}
     \end{center}
     \caption{Effect of changing {\swwc} buffer and partitioning fan-out settings on the performance of {\lsj}.}
     \label{fig:sj_swwc_partitioning_fanout}
\end{figure}

\noindent\textbf{{\swwc} Buffer.}  Figure~\ref{fig:sj_swwc_partitioning_fanout}(a) shows the effect of enabling the {\swwc} buffering on the performance of {\lsj} while joining different datasets. For synthetic datasets ({\seqh} and {\unif}), we show the results of joining the 640Mx640M variant. We can see that most datasets benefit from enabling buffering in the partitioning phase. However, such benefit is capped by the contribution of partitioning time to the overall {\sj} time (typically, the partitioning phase ranges from 20\% to 35\% of the total time). For these datasets, except {\seqh}, the throughput improvement ranges from 10\% to 30\%.

\noindent\textbf{Partitioning Fan-out.} Figure~\ref{fig:sj_swwc_partitioning_fanout}(b) shows the effect of choosing different fan-out values (2000, 10000, 40000, and 60000) on the throughput of {\blsj}, while joining the 128Mx128M variant of {\seqh} and {\lognorm} datasets. The results show that there is a sweet spot where the fan-out value hits the highest join throughput. This is expected because having a small fan-out will keep the number of comparisons in AVX sort large. In contrast, having a large fan-out will increase the TLB pressure when accessing the output partitions.

\begin{figure}
     \begin{center}
     \includegraphics[width=3.2in]{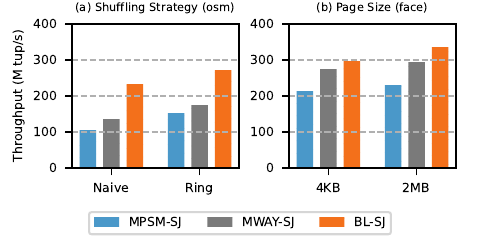}
     \end{center}
     \caption{\textcolor{black}{Effect of changing data shuffling strategy (NUMA) and page sizes (TLB entries) on the performance of the {\sj} algorithms.}}
     \label{fig:sj_numa_page_size}
\end{figure}

\noindent\textcolor{black}{\textbf{Data Shuffling.} In this experiment, we study how equipping the {\sj} algorithms with NUMA awareness can affect the join throughput. A recent study~\cite{LPM+13} showed that the implementation details of data shuffling across different NUMA nodes can substantially impact the interconnect bandwidth, and hence the overall performance. Since data shuffling is an essential operation in any {\sj} algorithm, optimizing the shuffling strategy is expected to improve the join throughput. Basically, both {\blsj} and {\mpsmsj} perform data shuffling during their partitioning and joining phases, while the shuffling in {\mwaysj} occurs in its partitioning and multi-way merge phases. Here, we experiment with two data shuffling strategies: (1)~\textit{Naive} shuffling (the default strategy in all algorithms), which does not consider the NUMA characteristics and lets each thread pull/write  data generated by all other threads without any optimization, and (2)~\textit{Ring} shuffling, which executes the shuffling operation on multiple steps, as described in~\cite{LPM+13}, to reduce the NUMA interconnect contention. Figure~\ref{fig:sj_numa_page_size}(a) shows how the join throughputs of different {\sj} algorithms change under these two shuffling strategies while using the {\osm} dataset. As expected, all algorithms benefit from the "Ring" shuffling strategy, where the improvement ratios of join throughput are 17\%, 30\% and 45\% for {\blsj}, {\mwaysj} and {\mpsmsj}, respectively. As observed, {\blsj} has the least improvement ratio. Generally speaking, the effect of using the "Ring" strategy becomes significant if it is used in a time-consuming phase, and vice versa. Therefore, since the partitioning and joining phases of {\blsj} do not represent more than 15\% of the total algorithm time, the {\blsj}'s throughput improvement due to using the "Ring" strategy becomes limited compared to {\mpsmsj} and {\mwaysj}.}

\noindent\textcolor{black}{\textbf{Page Size.} Figure~\ref{fig:sj_numa_page_size}(b) shows the effect of changing the page size from 4KB (the default value in Unix operating systems) to a larger size 2MB (a.k.a huge page) while using the {\fb} dataset. Compared to the {\inlj} algorithms, {\sj} algorithms perform larger amount of random memory accesses (specially when having big working sets, i.e., large datasets), which can be limited by the TLB misses. We can reduce such misses by utilizing larger page sizes as suggested in previous studies (e.g.,~\cite{SCD16}). As shown in the figure, all {\sj} algorithms benefit from increasing the page size. Specifically, the throughputs of {\blsj}, {\mwaysj} and {\mpsmsj} improved by 13\%, 7\%, 8\%, respectively. We can observe that the improvement ratio in {\blsj} is higher than in other algorithms. This is because the overhead of the partitioning phase in {\blsj} is a bit higher (check Section~\ref{sec:lsj_partition_phase}), and hence reducing the TLB misses will have a greater impact. We skipped reporting the results using larger page sizes (e.g., 1GB), because there were no throughput improvement as the working set already fits in the TLB when using 2MB pages.}




\section{Discussion}
\label{sec:discussion}
In this section, we discuss some performance aspects and limitations related to the learned in-memory joins.

\noindent\textbf{Duplicates.} Having duplicate keys in the input datasets results in building imprecise {\cdf}-based models that will lead to predicting the same position/partition for multiple keys. This increases the number of collisions, and hence reduces the join throughput (e.g., in {\inlj}, collisions increase the "last-mile" search overhead, while in {\lsj}, they result in imbalanced partitions that increase the sorting and joining latency). Although our experimental results (Figure~\ref{fig:all_duplicates}) show the competitiveness of learned {\inlj} and {\sj} variants when having a high number of duplicates, their performance gain over non-learned variants becomes insignificant at extreme duplicate ratios (e.g., 80\% or more). As a future work, we plan to employ efficient ideas from~\cite{KVK21} to further reduce the effect of duplicates.

\noindent\textbf{Complex Distributions.} In general, learned joins are favorable when inputs datasets have skewed and/or complex distributions (e.g., {\lognorm}, {\osm} and {\wiki} datasets). That being said, if the datasets never preserve local structure (e.g., high randomization and spikes), then the distributions of such datasets are difficult to model. In these cases, learned join variants require a significant large number of models (i.e., higher storage overhead and cache misses) to achieve errors comparable to those observed on the easy-to-learn datasets.

\noindent\textbf{Vectorization.} {\simd} has been heavily studied in the data management field to accelerate different operations, including scanning~\cite{LMF+16, DKF+18}, bloom filters~\cite{PR14_2}, non-learned joins~\cite{FZW19, FHC+19, LPK+20} and sorting~\cite{CNL+08, IT15}. A recent study~\cite{SVH+21} showed that, with the help of vectorization and prefetching-optimized inter-task parallelism (e.g., {\amac}~\cite{KFG15}), the prediction time of learned models can be further improved for indexing applications. We expect that employing the vectorization in learned joins will be beneficial as well, and investigating it is an interesting future research direction.

\section{Related Work}
\label{sec:relatedwork}

\noindent \textbf{Nested Loop Joins (NLJ).} Due to their simplicity and good performance in many cases, nested loop joins are widely used in database engines for decades~\cite{HSH07}. They have been extensively studied against other join algorithms either on their own (e.g.,~\cite{CZ14}) or within the query optimization context (e.g.,~\cite{LGM+15}). Block and indexed nested loop joins are the most popular variants~\cite{HL06}. Several works revisited their implementations to improve the performance for hierarchical caching (e.g.,~\cite{HL06}) and modern hardware (e.g., GPU~\cite{HYF+08}).

\noindent \textbf{Sort-based Joins ({\sj}).} Numerous algorithms have been proposed to improve the performance of sort-based joins over years. However, we here highlight the parallel approaches as they are the most influential ones~\cite{KKL+09, AKN12, BAT+13, SCD16, BMS+17}.~\cite{KKL+09} provided the first evaluation study for sort-based join in a multi-core setting, and developed an analytical model for the join performance that predicted the cases where sort-merge join would be more efficient than hash joins. 
After that, {\mpsm}~\cite{AKN12} was proposed to significantly improve the performance of sort-based join on {\numa}-aware systems.
Multi-Way sort-merge join (we refer to it as {\mway}) is another state-of-the-art algorithm~\cite{BAT+13} that improved over~\cite{KKL+09} by using wider {\simd} instructions and range partitioning to provide efficient multi-threading without heavy synchronization. 

\noindent \textbf{Hash-based Joins ({\hj}).} In general, hash-based joins are categorized into non-partitioned (e.g.,~\cite{TAB+13, BLP11, LLA+15}) and partitioned (e.g.,~\cite{TAB+13, SCD16, BGN21}) hash joins. Radix join~\cite{BMK99, MBK00, MBK02, TAB+13, BTA+15} is the most widely-used partitioned hash join due to its high throughput and cache awareness (e.g., avoiding TLB misses).~\cite{KKL+09, BLP11} provided first parallel designs for radix hash joins, while~\cite{TAB+13} and~\cite{BGN21} added software write-combine buffers ({\swwc}) and bloom filter optimizations, respectively, to these implementations. Other works optimized the performance of hash join algorithms for {\numa}-aware systems (e.g.,~\cite{LLA+15, SCD16}), hierarchical caching (e.g.,~\cite{HL07}), and special hardware (e.g.,GPU~\cite{HYF+08}, FPGA~\cite{CCB+20}, NVM~\cite{STC+20}).





\noindent \textbf{Machine Learning for Database.} During the last few years, machine learning started to have a profound impact on automating the core database functionality and design decisions. Along this line of research, a corpus of works studied the idea of replacing traditional indexes with learned models that predict the location of a key in a dataset including single-dimension (e.g.,~\cite{KBC+18, KMR+20, FV20}), multi-dimensional (e.g.,~\cite{NDA+20, DNA+20}), updatable (e.g.,~\cite{DMY+20}), and spatial (e.g.,~\cite{LLZ+20, QLJ+20, PRK+20}) indexes.
Query optimization is also another area that has several works on using machine learning to either entirely build an optimizer from scratch (e.g.,~\cite{MNM+19}), or provide an advisor that improves the performance of existing optimizers (e.g.,~\cite{MNM+21, SLM+01, VZZ+00}). Moreover, there exists several workload-specific optimizations, enabled by machine learning, in other database operations e.g., cardinality estimation~\cite{KKR+19}, data partitioning~\cite{YCW+20}, sorting~\cite{KVC+20}, and scheduling~\cite{MSV+19}. However, to the best of our knowledge, there is no existing work on using the machine learning to improve in-memory joins. Our proposed work in this paper fills this gap. 


\section{Conclusion}
\label{sec:conclusion}

In this paper, we discussed and experimentally investigated the usage of {\cdf}-based partitioning and learned indexes (e.g., {\rmi} and {\rs}) in the three join categories; indexed nested loop join ({\inlj}), sort-based joins ({\sj}) and hash-based joins ({\hj}). We found that learned models can not improve the performance of {\hj}, yet, they are beneficial for {\inlj} and {\sj}. However, directly integrating learning models with {\inlj} and {\sj} leads to a sub-par performance. Therefore, we proposed optimized {\inlj} and {\sj} variants that properly exploit these learned models. In {\inlj} category, we introduced learned {\inlj}s that employ efficient \textit{model-based insertion} and \textit{cache-optimized buffering} techniques. In {\sj} category, we proposed a new sort-based join, namely {\lsj}, that exploits a {\cdf}-based model to improve the different {\sj} phases (e.g., sort and join) by just scaling up its prediction degree. Our experimental evaluation showed that learned {\inlj} can be 2.7X faster than hash-based {\inlj} in many scenarios and with different datasets. In addition, the proposed {\lsj} can be 2X faster than Multi-way sort-merge join~\cite{BAT+13}, which is the best sort-based join baseline in the literature.


\newpage

\bibliographystyle{plain} 
\bibliography{ml_databases}

\end{document}